\documentclass[a4paper,12pt,english]{scrartcl}

\usepackage[ansinew]{inputenc}
\usepackage[T1]{fontenc}
\usepackage{lmodern}
\usepackage{amsmath}
\usepackage{amsthm}
\usepackage[english]{babel}
\usepackage{amsfonts}
\usepackage{amssymb}
\usepackage{array} 
\usepackage{float}
\usepackage[dvipsnames]{xcolor} 
\usepackage{enumitem}
\usepackage{pgf,tikz}
\usetikzlibrary{shapes}
\usetikzlibrary{shapes.misc}
\usepackage[colorlinks = true,
            linkcolor = blue,
            urlcolor  = blue,
            citecolor = blue]{hyperref}

\newtheorem{theorem}[]{Theorem}[section]
\newtheorem{definition}[theorem]{Definition}

\newtheorem{corollary}[theorem]{Corollary}
\newtheorem{proposition}[theorem]{Proposition}
\newtheorem{example}[theorem]{Example}
\newtheorem{algo}[theorem]{Algorithm}

\newtheoremstyle{mycases}{}{}{}{}{}{}{\newline}{\textbf{\thmname{#1}\thmnumber{ #2:}}\thmnote{ \textit{#3}}}
\theoremstyle{mycases}

\newlist{caselist}{enumerate}{1}
\setlist[caselist]{label=\bfseries Case \arabic*:,ref=\bfseries Case \arabic*,labelindent=-1.5cm,leftmargin=*,widest="Case 6"}

\newcommand{\N}{\mathbb N} 
\newcommand{\Z}{\mathbb Z} 

\newcommand{\K}{\mathbb K} 

\newcommand{\sst}[2]{\left\lbrace #1\, \textnormal{\textbf{\textbar}}\,#2 \right\rbrace} 

\newcommand{\ie}{i.\,e.\ }

\newcommand{\eqspace}{\,}

\DeclareMathOperator{\ord}{ord}

\title{Rational Solutions of High-Order Algebraic Ordinary Differential Equations}
\author{Thieu N. Vo\thanks{Faculty of Mathematics and Statistics, Ton Duc Thang University, Ho Chi Minh City, Vietnam. 
Email: vongocthieu@tdt.edu.vn}
	\, and Yi Zhang\thanks{Johann Radon Institute for Computational and Applied Mathematics (RICAM), Austrian Academy of Sciences, Austria. 
	Supported by the Austrian Science Fund (FWF): P29467-N32. Email: zhangy@amss.ac.cn}
}
\date{\today}

\begin{document}
\maketitle
\begin{abstract}
We consider algebraic ordinary differential equations (AODEs) and study their polynomial and rational solutions.
A sufficient condition for an AODE to have a degree bound for its polynomial solutions is presented.
An AODE satisfying this condition is called \emph{noncritical}.
We prove that usual low order classes of AODEs are noncritical.
For rational solutions, we determine a class of AODEs, which are called \emph{maximally comparable}, 
such that the poles of their rational solutions are recognizable from their coefficients.
This generalizes a fact from linear AODEs, that the poles of their rational solutions are the zeros of the corresponding
highest coefficient.
An algorithm for determining all rational solutions, if there is any, of certain maximally comparable AODEs, 
which covers $78.54\%$ AODEs from a standard differential equations collection by Kamke, is presented.

\end{abstract}

\section{Introduction}

An algebraic ordinary differential equation (AODE) is of the form 
$$F(x,y,y',\ldots,y^{(n)})=0,$$ 
where $F$ is a polynomial in $y,y',\ldots,y^{(n)}$ with coefficients in $\mathbb{K}(x)$, the field of rational functions 
over an algebraically closed field $\K$ 
of characteristic zero, and $n \in \N$. 
For instance, $\K$ can be the field of complex numbers, 
or the field of algebraic numbers. 
Many problems from applications (such as physics, combinatorics and statistics) 
can be characterized in terms of AODEs.  
Therefore, determining (closed form) solutions of an AODE is one of the central problems in mathematics and computer science.
 
Although linear ODEs~\cite{Ince1926} have been intensively studied, there are still many challenging problems for solving
(nonlinear) AODEs. As far as we know, approaches for solving AODEs are only available for very specific subclasses. 
For example, Riccati equations, 
which have the form $y'=f_0(x)+f_1(x)y+f_2(x)y^2$ for some $f_0,f_1,f_2 \in \mathbb{K}(x)$, 
can be considered as the simplest form of nonlinear AODEs. 
In \cite{Kovacic}, Kovacic gives a complete algorithm for determining Liouvillian solutions of 
a Riccati equation with rational function coefficients. The study of general solutions without movable singularities 
can be found in \cite{Fuchs, Malmquist, Poincare} for first-order, and in \cite{Eremenko2, Ince1926} for higher-order AODEs. 

Since the problem of solving an arbitrary AODE is very difficult, 
it is natural to ask whether a given AODE admits some special kinds of solutions. 
We are interested in polynomials and rational functions. 
During the last two decades, an algebraic-geometric approach for finding rational solutions of AODEs has been developed. 
In \cite{Eremenko}, Eremenko gave a theoretical consideration for the existence of a degree bound for rational solutions of a first-order AODE.
In \cite{FengGao, FengGao06}, by using a view from algebraic curves, the authors provided polynomial time algorithms for determining rational 
(and algebraic) solutions of an autonomous first-order AODE.
The authors of~\cite{NgoWinkler11b, VoWinkler2015, VoGraseggerWinkler2017}
developed the methods for non-autonomous first-order AODEs. 

In this paper, we are interested in polynomial and rational solutions of arbitrary order AODEs and their properties.
We give a sufficient condition for an AODE to have a degree bound for its polynomial solutions,
and in the affirmative case, determine such a bound.
An AODE satisfying this condition is called \emph{noncritical}.
The easy determination of the condition allows us to confirm that several usual low order AODEs 
are noncritical (see Theorem~\ref{PROP:noncritical1} and~\ref{PROP:noncritical2}).
This result can be considered as a refinement of the works of polynomial solutions of Krushel'nitskij in~\cite{Krushel}, 
and Cano in~\cite{Cano2005}.

It well-known that every pole of rational solutions of a linear ODE with polynomial coefficients is a zero point 
of its highest coefficient.
This fact allows us to easily recognize possible poles of a rational solutions from the coefficients of a given linear AODE.
Unfortunately, this fact is no longer true for nonlinear AODEs. 
However, we show that there is a big subclass of AODEs in which this fact is still valid.
In order to do that, we equip the set of monomials in the unknown $y$ and its derivatives with a suitable partial order 
(see Definition~\ref{DEF:partialord}). If an AODE admits the highest monomial with respect to this ordering, 
then poles of its rational solutions can only occur at the zeros of the corresponding highest coefficient (Theorem~\ref{thm:movpol}).
This generalizes the same fact of linear AODEs to the nonlinear ones.
An AODE satisfying the existence of the highest monomial is called \emph{maximally comparable}.

The notion of maximally comparable AODEs already appears in \cite{Vo2018}, where the authors considered first-order AODEs only.
The authors proved that for every maximally comparable first-order AODE, there is a finite upper bound for the degrees of its rational solutions, 
together with an algorithm to determine the bound. 
Here, we extend the notion to high-order ones.
Unlike the first-order cases, there might be no such an upper bound for the higher order ones. 
We define a class of AODEs, called \emph{completely maximally comparable}, 
where the existence of an upper bound for its rational solutions is always guaranteed.
The class of maximally comparable AODEs covers $78.54\%$ AODEs from a standard collection by Kamke~\cite{Kamke}.
All of them are complete.
This suggests that completely maximally comparable AODEs, 
which are in the scope of our algorithm for determining all rational solutions (see Algorithm~\ref{ALGO:ratsol}), 
form a big subclass among AODEs.

The rest of the paper is organized as follows. Section \ref{sec:LaurentSeriesSolutions} 
is devoted for a study of order bounds for poles of a Laurent series solution of AODEs. 
In Section~\ref{sec:PolynomialSolutions}, we give a sufficient condition for an AODE to have a degree bound for its polynomial solutions.
We also prove that usual low order AODEs satisfy this condition.
Rational solutions of maximally comparable are considered in Section~\ref{sec:RationalSolutions}.
Finally, we perform a statistical investigation with a collection of AODEs from a standard text book by Kamke~\cite{Kamke}.

\section{An order bound for Laurent series solutions} \label{sec:LaurentSeriesSolutions}
 

This section can be considered as an alternative interpretation of the Newton polygon method for AODEs, specified for Laurent series solutions.
In particular, given an AODE, we show in Proposition~\ref{PROP:LaurentSeriesSols} 
that the orders of its Laurent series solutions at any point can be bounded in an algorithmic way.
The proposition yields an easy determination of the bound.
More general constructions which are applicable for wider 
classes of series solutions can be found in \cite{Cano2005,DoraJung1997,GrigorievSinger1991}.

Given $x_0 \in \mathbb{K} \cup \{\infty\}$, a Laurent series $f$ at $x=x_0$ has the 
form 
\[
\begin{array}{lll}
 \sum\limits_{k=m}^{\infty}{c_k(x-x_0)^k} &  & \ \ \text{if} \ \ x_0 \in \mathbb{K}, \\
 \sum\limits_{k=m}^{\infty}{c_kx^{-k}}    &  & \ \ \text{if} \ \ x_0 = \infty, \\
\end{array}
\]
where $c_k \in \mathbb{K}, c_m \neq 0$ and $m \in \Z$.
We call $-m$ \emph{the order} of $f$ (at $x = x_0$), and denote it by $\text{ord}_{x_0}(f)$. 
The coefficient $c_m$ is called \emph{the lowest coefficient} of $f$ (at $x = x_0$), 
and denoted by $c_{x_0}(f)$. 
Then we can rewrite~$f$ as follows:
\[
\begin{array}{lllll}
 c_{x_0}(f)(x-x_0)^{-\text{ord}_{x_0}(f)}  & + &  \text{ higher terms in } (x - x_0) &  & \ \ \text{if} \ \ x_0 \in \mathbb{K}, \\
 c_{\infty}(f)x^{\text{ord}_{x_0}(f)}  & + &  \text{ lower terms in } x    &  & \ \ \text{if} \ \ x_0 = \infty. \\
\end{array}
\]

For each $I = (i_0,i_1,\ldots,i_n) \in \mathbb{N}^{n+1}$ and $r \in \{ 0,\ldots,n \}$, 
we set $||I||_r=i_{r}+\ldots+i_{n}$. 
We simply write $||I||_0$ by $||I||$. 
Furthermore, the notation $||I||_\infty=i_1+2i_2+\ldots+ni_n$ will be also used frequently.


Let $F(y)=\sum\limits_{I \in \mathbb{N}^{n+1}}{f_{I}(x) y^{i_0} (y')^{i_1} \cdots (y^{(n)})^{i_n}} \in \mathbb{K}(x)\{y\}$ 
be a differential polynomial of order $n$. We will use the following notations:
\[
\begin{array}{lll}
\mathcal{E}(F) & = & \{I \in \mathbb{N}^{n+1} \,|\, f_{I} \neq 0\},\\
 d(F)           & = & \max \{||I|| \,|\, I \in \mathcal{E}(F)\},\\
\mathcal{D}(F) & = & \{I \in \mathcal{E}(F) \,|\, ||I||=d(F)\}.
\end{array}
\]
Moreover, for each $x_0 \in \mathbb{K}$, we denote 
\[
\begin{array}{lll}
m_{x_0}(F) & = & \max \{\ord_{x_0}f_{I}+||I||_\infty \,|\, I \in \mathcal{D}(F)\},\\
\mathcal{M}_{x_0}(F) & = &\{I \in \mathcal{D}(F) \,|\, \ord_{x_0}f_{I}+||I||_\infty=m_{x_0}(F)\},\\
\mathcal{P}_{x_0,F}(t) & = &\sum\limits_{I \in \mathcal{M}_{x_0}(F)}{c_{x_0}(f_{I}) \cdot \prod\limits_{r=0}^{n-1}{(-t-r)}^{||I||_{r+1}}},
\end{array}
\]
and if $\mathcal{E}(F) \setminus \mathcal{D}(F) \neq \emptyset$, we set
\begin{equation*}
b_{x_0}(F)= \max \left \{ \frac{\ord_{x_0} f_{I}+||I||_\infty-m_{x_0}(F)}{d(F)-||I||}  \mathrel{\Big |} 
I \in \mathcal{E}(F) \setminus \mathcal{D}(F) \right \}.
\end{equation*}

\noindent In case that $x_0 = \infty$, we also denote 
\[
\begin{array}{lll}
m_{\infty}(F) & = & \max \{\ord_{\infty}f_{I}-||I||_\infty \,|\, I \in \mathcal{D}(F)\},\\
\mathcal{M}_{\infty}(F) & = & \{I \in \mathcal{D}(F) \,|\, \ord_{\infty}f_{I} - ||I||_\infty=m_{\infty}(F)\},\\
\mathcal{P}_{\infty,F}(t) & = & \sum\limits_{I \in \mathcal{M}_{\infty}(F)}{c_{\infty}(f_{I}) \cdot \prod\limits_{r=0}^{n-1}{(t-r)}^{||I||_{r+1}}},
\end{array}
\]
and
\begin{equation*}
b_{\infty}(F)= \max \left \{ \frac{\ord_{\infty} f_{I}-||I||_\infty-m_{\infty}(F)}{d(F)-||I||} \mathrel{\Big |} 
I \in \mathcal{E}(F) \setminus \mathcal{D}(F) \right \}
\end{equation*}
if $\mathcal{E}(F) \setminus \mathcal{D}(F) \neq \emptyset$.

\begin{definition} \label{DEF:indicialpol}
Let $F(y) \in \K(x)\{y\}$ be a differential polynomial of order $n$. 
For each $x_0 \in \mathbb{K} \cup \{\infty\}$, 
we call $\mathcal{P}_{x_0,F}$ the \emph{indicial polynomial} of $F$ at $x=x_0$.
\end{definition}

Note that the above definition  is a generalization of 
the usual indicial polynomial~\cite{Yi2017, Ince1926} of linear ODEs.

\begin{proposition} \label{PROP:LaurentSeriesSols}
Given an AODE $F(y)=0$, and $x_0 \in \mathbb{K} \cup \{\infty\}$. 
If $r \geq 1$ is the order of a Laurent series solution of $F(y) = 0$ at $x=x_0$, 
then one of the following claims hold:
\begin{itemize}
\item[(i)] $\mathcal{E}(F) \setminus \mathcal{D}(F) \neq \emptyset$, and $r \leq b_{x_0}(F)$;
\item[(ii)] $r$ is a positive integer root of $\mathcal{P}_{x_0,F}(t)$.
\end{itemize}
\end{proposition}

\begin{proof}
Let $F(y)=\sum\limits_{I \in \mathbb{N}^{n+1}}{f_I(x)y^{i_0}(y')^{i_1} \ldots (y^{(n)})^{i_n}} \in \mathbb{K}(x)\{y\}$ be 
a differential polynomial of order $n$.
Let $x_0 \in \mathbb{K}$ and $z \in \mathbb{K}((x-x_0)) \setminus \mathbb{K}$ be 
a Laurent series solution of $F(y) = 0$ of order $r \geq 1$.
Then $z^{(k)}$ is of order $k + r$ for each $k \in \mathbb{N}$. 
For each $I \in \mathcal{E}(F)$, we may write the coefficient $f_I$ in the following form:
\[f_I = \frac{c_{x_0}(f_I)}{(x-x_0)^{\ord_{x_0}f_I}} + h_I, \]
where $h_I \in \mathbb{K}((x))$ and $\ord_{x_0}{h_I}<\ord_{x_0}f_{I}$. 
Since $z$ is a solution of $F(y) = 0$, we have

\[
\begin{array}{lll}
 0 & = & F(z) \\
   & = & S_1 + S_2 + S_3 + S_4,
\end{array}
\]
where 
\[
\begin{array}{lll}
 S_1 = \sum\limits_{I \in \mathcal{M}_{x_0}(F)}{\frac{c_{x_0}(f_I)}{(x-x_0)^{\ord_{x_0}f_I}} \cdot z^{i_0} (z')^{i_1} \cdots (z^{(n)})^{i_n}}, 
 &  & S_2 =  \sum\limits_{I \in \mathcal{M}_{x_0}(F)}{h_I \cdot z^{i_0} (z')^{i_1} \cdots (z^{(n)})^{i_n}},\\
 S_3 = \sum\limits_{I \in \mathcal{D}(F) \setminus \mathcal{M}_{x_0}(F)}{f_I  z^{i_0} (z')^{i_1} \cdots (z^{(n)})^{i_n}}, 
 &  &  S_4 = \sum\limits_{I \in \mathcal{E}(F) \setminus \mathcal{D}(F)}{f_I  z^{i_0} (z')^{i_1} \cdots (z^{(n)})^{i_n}}. 
\end{array}
\]
The order of each term in $S_1$ are equal to $D= d(F) r+m_{x_0}(F)$, which is strictly larger than 
that of each term in $S_2$ and $S_3$. 
One of the two following cases will happen:
\begin{caselist}
\item The order of $S_1$ is equal to $D$. 
Then the term of order $D$ in $S_1$ must be killed by terms of $S_4$.
In this case, we have $\mathcal{E}(F) \setminus \mathcal{D}(F) \neq \emptyset$. 
By comparing with the orders of terms in $S_4$, we obtain
\begin{equation*}
D \leq \max \sst{||I|| \cdot r+||I||_\infty+\ord_{x_0} f_I}{I \in \mathcal{E}(F) \setminus \mathcal{D}(F)} \eqspace.
\end{equation*}
On the other hand, since $D = d(F) r+m_{x_0}(F)$, we conclude that
\begin{equation*}
r \leq \max\left\{\frac{||I||_\infty+\ord_{x_0} f_I-m_{x_0}(F)}{d(F)-||I||} \mathrel{\Big |} I \in \mathcal{E}(F) \setminus \mathcal{D}(F) \right\}.
\end{equation*}
In other words, $r \leq b_{x_0}(F)$.
\item The order of $S_1$ is strictly smaller than $D$. 
For each $k \in \mathbb{N}$, a direct computation implies that the lowest coefficient $z^{(k)}$ at $x = x_0$ is
\[
c_{x_0}(z^{(k)}) = c_{x_0}(z) \prod\limits_{s=1}^{k}{(-r-s+1)}.
\]
Therefore, the lowest coefficient of the term indexed by $I \in \mathcal{M}_{x_0}(F)$ in $S_1$ is 
\begin{align*}
c_{x_0}(f_I) \cdot \prod\limits_{k=0}^{n} \left({c_{x_0}(z) \prod\limits_{s=1}^{k}{(-r-s+1)}} \right)^{i_k} 
= c_{x_0}(f_I) c_{x_0}(y)^{||I||} \prod\limits_{s=1}^{n}{(-r-s+1)^{||I||_s}}. 
\end{align*}
Since the orders of terms in $S_1$ are the same and they are strictly larger than that of $S_1$, 
the sum of those lowest coefficients must be zero. 
In other words, we have 
\begin{align*}
\sum\limits_{I \in \mathcal{M}_{x_0}(F)}{ c_{x_0}(f_I) c_{x_0}(y)^{||I||} \prod\limits_{s=1}^{n}{(-r-s+1)^{||I||_s}} }=0. 
\end{align*}
The left side of the above equality is exactly $c_{x_0}(y)^{d(F)} \cdot \mathcal{P}_{x_0,F}(r)$. 
Hence, $r$ is a positive integer root of $\mathcal{P}_{x_0,F}(r)$.
\end{caselist}

The case that $x_0 = \infty$ can be proved in a similar way.
\end{proof}


For a linear homogeneous ordinary differential equation $F(y) = 0$, 
item (i) of the above theorem will never happen because $\mathcal{E}(F) = \mathcal{D}(F)$.

Consider an AODE $F(y) = 0$. If $\mathcal{E}(F) = \mathcal{D}(F)$ and the indicial polynomial $\mathcal{P}_{x_0,F}(t)$ is zero, 
then Theorem~\ref{PROP:LaurentSeriesSols} does not give any information for the order bound of Laurent series solution of $F(y) =0$ 
at $x = x_0$. In the next section, we will give an example (Example~\ref{exa:Nondicritical}) 
that the order can be arbitrarily high in this case. 

\section{Polynomial solutions of noncritical AODEs} \label{sec:PolynomialSolutions}


In \cite{Krushel}, Krushel'nitskij discusses the properties of the degree of a polynomial solution for a given AODE. 
By using the Newton polygon at infinity, Cano proposes an algorithm for determining a bound for 
the degrees of polynomial solutions of an AODE provided that the Newton polygon of the given AODE must satisfies certain additional assumption (see \cite[Section~2.2]{Cano2005}).
Whenever a degree bound is found, one can determine all polynomial solutions by undeterminate coefficient method.
However, to the best of our knowledge, no full algorithm for computing all polynomial solutions of AODEs exists so far.

In this section, we use Proposition~\ref{PROP:LaurentSeriesSols} to give a sufficient condition (Definition~\ref{def:Noncritical}) 
for the existence of a bound 
for the degrees of polynomial solutions.
We prove that several usual classes of AODEs satisfy this sufficient condition (Theorem~\ref{PROP:noncritical1} and Theorem~\ref{PROP:noncritical2}).
Furthermore, we will show in Section~\ref{sec:ExperimentalResults} that all of AODEs 
in Kamke's collection~\cite{Kamke} satisfy the sufficient condition.


\begin{definition} \label{def:Noncritical}
An AODE $F(y)=0$ is called \emph{noncritical} if $\mathcal{P}_{\infty,F}(t) \neq 0$.
\end{definition}


\begin{corollary} \label{COR:noncritical}
If an AODE $F(y)=0$ is noncritical, then there exists a bound for the degree of its polynomial solutions.
\end{corollary}

\begin{proof}
Straightforward from Theorem~\ref{PROP:LaurentSeriesSols}. 
\end{proof}

\begin{algo} \label{ALGO:polsol}
Given a noncritical AODE $F(y) = 0$, compute all its polynomial solutions.
\begin{itemize}
\item [(1)] Compute $\mathcal{P}_{\infty,F}(t)$. 
If $\mathcal{P}_{\infty,F}(t)$ has integer roots, then set $r_1$ to be the largest integer root. 
Otherwise, set $r_1 = 0$.
\item [(2)] Compute $r_2 = \lfloor b_{\infty}(F) \rfloor$ if $\mathcal{E}(F) \setminus \mathcal{D}(F) \neq \emptyset$. 
Otherwise set $r_2 = 0$.
\item [(3)] Set $r = \max\{r_1, r_2, 0\}$. Make an ansatz $z = \sum_{i = 0}^r c_i x^i$, where $c_i$'s are unknown. 
Substitute $z$ into $F(y) = 0$ and solve the corresponding algebraic equations by using Gr\"{o}bner bases. 
\item [(4)] Return the solutions from the above step. 
\end{itemize}
\end{algo}


The termination of Algorithm \ref{ALGO:polsol} is obvious. 
The correctness follows from Theorem~\ref{PROP:LaurentSeriesSols}.

\begin{example}[Kamke 6.234 \cite{Kamke}]
Consider the differential equation:
\begin{equation} \label{exa:findPolynomialSolutions}
F(y)=a^2y^2y''^2-2a^2yy'^2y''+a^2y'^4-b^2y''^2-y'^2=0,
\end{equation}
where $a,b \in \mathbb{K}$ and $a \neq 0$. 
The following table is a list of the exponents of terms of $F$ and related information.

\begin{center}
\begin{tabular}{|c c c c|}
\hline
$I \in \mathcal{E}(F)$ & $||I||$ & $||I||_\infty$ & $f_I$\\
\hline
$(2,0,2)$ & $4$ & $4$ & $a^2$\\
$(1,2,1)$ & $4$ & $4$ & $-2a^2$\\
$(0,4,0)$ & $4$ & $4$ & $a^2$\\
\hline
$(0,0,2)$ & $2$ & $4$ & $-b^2$\\
$(0,2,0)$ & $2$ & $2$ & $-1$\\
\hline
\end{tabular}
\end{center}

From the above table we see that $\mathcal{D}(F)$ is the set of exponents in the first three lines, 
and $\mathcal{E}(F) \setminus \mathcal{D}(F)$ is the set of exponents in the last two lines. 
A direct computation shows that $m_\infty(F)=-4$, $\mathcal{M}_{\infty}(F)=\mathcal{D}(F)$, 
and $\mathcal{P}_{\infty,F}(t)=a^2t^2 \neq 0$. 
Therefore, the differential equation~\eqref{exa:findPolynomialSolutions} is noncritical. 
Furthermore, we find that $b_{\infty}(F)=1$. 

By Theorem~\ref{PROP:LaurentSeriesSols}, 
every polynomial solution of~\eqref{exa:findPolynomialSolutions} has degree at most 1. 
By making an ansatz and solving the corresponding algebraic equations, 
we obtain all polynomial solutions, 
which are $c$, $c+\frac{x}{a}$, and $c-\frac{x}{a}$, where $c$ is an arbitrary constant in $\K$.
\end{example}

Through our investigation, almost all AODEs we see in the literature are noncritical (see Section \ref{sec:ExperimentalResults}). 
Only few of them are not noncritical. Below is one example for a critical AODE. 

\begin{example} \label{exa:Nondicritical}
Consider the following differential equation~\cite{GrigorievSinger1991,Eremenko}:
$$F(y) = xyy''-xy'^2+yy'=0.$$ 
By computation, we find that its indicial polynomial is zero. 
So, $F(y) = 0$ is a critical AODE.
Actually, it has polynomial solutions $z = cx^n$ for arbitrary $c \in \mathbb{K}$ and $n \in \mathbb{N}$.
\end{example}

We show in the next two theorems that noncritical AODEs cover most of usual low order AODEs. 

\begin{theorem} \label{PROP:noncritical1}
Let $\mathcal{L} \in \mathbb{K}(x) \left[ \frac{\partial}{\partial x} \right]$ be a differential operator, 
and $P(x,y,z) \in \mathbb{K}(x)[y,z]$ a polynomial in two variables with coefficients in $\mathbb{K}(x)$. 
Then for each $n > 0$, the differential equation $\mathcal{L}(y)+P(x,y,y^{(n)})=0$ is noncritical.

In particular, linear AODEs, first-order AODEs (which have the form $F(x,y,y')=0$ for some $F \in \mathbb{K}(x)[y,y']$), 
and quasi-linear second-order AODEs 
(which have the form $y''+G(x,y,y')=0$ for some $G \in \mathbb{K}(x)[y,y']$), are noncritical.
\end{theorem}

\begin{proof}
Let $F(y)=\mathcal{L}(y)+P(x,y,y^{(n)})$. 
We prove that $\mathcal{P}_{\infty,F}$ is nonzero.

First, we consider the case that $P$ is a linear polynomial in $y$ and $z$. Then $F$ is a linear differential polynomial, say
$$F(y)=f_{I_{-1}}+f_{I_{0}}y+\cdots+f_{I_m}y^{(m)},$$
where $f_{I_i} \in \mathbb{K}(x)$ and $f_{I_m} \neq 0$ and $m \in \mathbb{N}$. 
A direct computation shows that the indicial polynomial of $F$ 
at infinity is of the form
$$\mathcal{P}_{\infty,F}(t) =
\sum_{\substack{i=0,\ldots,m \\ I_i \in \mathcal{M}_{\infty}(F)}}{c_{\infty}(f_{I_i}) \cdot \prod_{s = 1}^i (t - s + 1)},$$
which is a nonzero polynomial. Therefore, linear AODEs are noncritical.

Next, assume that $P$ is of total degree at least $2$. Then 
we have $\mathcal{D}(F)=\mathcal{D}(P(x,y,y^{(n)}))$ and $\mathcal{M}_{\infty}(F)=\mathcal{M}_{\infty}(P(x,y,y^{(n)}))$. 
We write $P(x,y,y^{(n)})$ in the form
$$P(x,y,y^{(n)})=\sum\limits_{(i,j) \in \mathbb{N}^2}{f_{i,j}(x)y^i(y^{(n)})^j}.$$
Then $\mathcal{M}_{\infty}(F)$ consists of elements of the 
form $e_{i,j}=(i,0,\ldots,0,j) \in \mathbb{N}^{n+1}$. A direct calculation reveals that  
$$\mathcal{P}_{\infty,F}(t)=
\sum_{\substack{j=1,\ldots,n \\ e_{i,j} \in \mathcal{M}_{\infty}(F)}} {c_{\infty}(f_{i,j}) \cdot \left[ t(t-1) \cdots (t-n+1) \right]^j}.$$
The indicial polynomial $\mathcal{P}_{\infty,F}(t)$ can be viewed as the evaluation of the 
nonzero univariate polynomial 
$$g(T) = \sum\limits_{\substack{j=1,\ldots,n \\ e_{i,j} \in \mathcal{M}_{\infty}(F)}} {c_{\infty}(f_{i,j}) \cdot T^j} 
\ \ \text{ at } \ \ T = t(t-1) \cdots (t-n+1).$$ 
On the other hand, since $t(t-1) \cdots (t-n+1)$ is transcendental over $\K$, 
we conclude that $\mathcal{P}_{\infty,F} \neq 0$.
\end{proof}

\begin{theorem} \label{PROP:noncritical2}
Let $\mathcal{L} \in \mathbb{K}(x) \left[ \frac{\partial}{\partial x} \right]$ be a differential operator with coefficients in $\mathbb{K}(x)$, 
and $Q(y,z,w) \in \mathbb{K}[y,z,w]$ a polynomial in three variables with coefficients in $\mathbb{K}$. 
Then for each $m, n > 0$, the differential equation $\mathcal{L}(y)+Q(y,y^{(n)},y^{(m)})=0$ is noncritical.

In particular, autonomous second-order AODEs (which have the form $F(y,y',y'')=0$ for some $F \in \mathbb{K}[y,y',y'']$), 
and quasi-linear autonomous third-order AODEs (which have the form $y'''+G(y,y',y'')=0$ for some $G \in \mathbb{K}[y,y',y'']$), are noncritical.
\end{theorem}

\begin{proof}
Let $F(y)=\mathcal{L}(y)+Q(y,y^{(m)},y^{(n)})$. Without loss of generality, we can assume that $0<m<n$. 
As we have seen from the previous proposition, a linear AODE is noncritical. 
Therefore we can assume further that $Q$ is of total degree at least $2$. 
Then we have $\mathcal{D}(F)=\mathcal{D}(Q(y,y^{(m)},y^{(n)}))$ and $\mathcal{M}_{\infty}(F)=\mathcal{M}_{\infty}(Q(y,y^{(m)},y^{(n)}))$. 
Let us write $Q(y,y^{(m)},y^{(n)})$ in the form
$$Q(y,y^{(m)},y^{(n)})=\sum\limits_{(ijk) \in \mathbb{N}^3} {f_{ijk}y^i(y^{(m)})^j (y^{(n)})^k}.$$
For simplicity, we denote $e_{ijk}=(i,0,\ldots,0,j,0,\ldots,0,k) \in \mathbb{N}^{n+1}$, 
where $j$ is the $(m+1)$-th coordinate. 
Then $\mathcal{M}_{\infty}(F)$ consists of all $e_{ijk}$ such that $i+j+k=d(F)$ and $mj+nk=m_{\infty}(F)$. 
A direct computation implies that  
\begin{align*}
\mathcal{P}_{\infty,F}(t) = \sum\limits_{\substack{(i,j,k) \in \mathbb{N}^3 \\ e_{ijk} \in \mathcal{M}_{\infty}(F)}} {c_{\infty}(f_{ijk}) \cdot 
\left( t(t-1)\cdots (t-m+1) \right)^{j+k} \cdot \left( (t-m) \cdots (t-n+1) \right)^k}.
\end{align*}
This polynomial can be rewritten as:
\begin{equation} \label{EQ:indicialpol2}
\mathcal{P}_{\infty,F}(t) = A^{\frac{m_\infty(F)}{m}} \cdot 
\sum\limits_{\substack{k=0,\ldots,n \\ e_{ijk} \in \mathcal{M}_{\infty}(F)}} {c_{\infty}(f_{ijk}) \left( \frac{B}{A^{\frac{(n-m)}{m}}} \right)^k}, 
\end{equation}
where $A=t(t-1)\cdots (t-m+1)$ and $B=(t-m) \cdots (t-n+1)$. 
The sum in~\eqref{EQ:indicialpol2} can be viewed as the evaluation of the univariate polynomial 
\[
 h(T) = \sum\limits_{\substack{k=0,\ldots,n \\ e_{ijk} \in \mathcal{M}_{\infty}(F)}} {c_{\infty}(f_{ijk}) T^k} 
 \ \ \text{ at } \ \ T=\frac{B}{A^{\frac{(n-m)}{m}}}.
\]
Since the projection which maps $e_{ijk}$ to $k$ is injective, we have that $h(T)$ is nonzero. 
On the other hand, since $\frac{B}{A^{\frac{(n-m)}{m}}}$ is transcendental over $\K$, 
we conclude that $\mathcal{P}_{\infty,F}$ is nonzero.
\end{proof}


\section{Rational solutions of maximally comparable AODEs} \label{sec:RationalSolutions}




It is well-known that poles of rational solutions of a linear ODE with polynomial coefficients 
only occur at the zeros of the highest coefficient of the equation (see~\cite{Ince1926}).
This fact is no longer correct when we consider nonlinear AODEs in general.
In this section, we describe a class of AODEs in which the above fact is still true. 
In order to do that, we first need to define what is the "highest" coefficient in the nonlinear case. 
To do so, we equip the set of monomials in $y$ and its derivatives with a suitable partial order (Definition~\ref{DEF:partialord}). 
We show in Theorem~\ref{thm:movpol} that if the given AODE has the greatest monomial with respect to this ordering, 
then the poles of its rational solutions can only occur at the zeros of the corresponding coefficient. 
Together with Proposition~\ref{PROP:LaurentSeriesSols}, we give a sufficient condition for such AODEs to have bounds for the orders of their poles, 
therefore one can determine their rational solutions if there is any.

\begin{definition} \label{DEF:partialord}
Assume that $n \in \N$.
For each $I,J \in \N^{n+1}$,
we say that $I \gg J$ 
if $||I|| \geq ||J||$ and $||I||+||I||_\infty > ||J||+||J||_\infty$.
\end{definition}

It is straightforward to verify that the order defined as above is a strict partial ordering on $\N ^{n+1}$, 
\ie the following properties hold for all $I,J,K \in \N ^{n+1}$:
\begin{enumerate}\renewcommand{\theenumi}{\roman{enumi}}\renewcommand{\labelenumi}{(\theenumi)}
\item irreflexivity: $I \not\gg I$;
\item transitivity: if $I \gg J$ and $J \gg K$, then $I \gg K$;
\item asymmetry: if $I \gg J$, then $J \not\gg I$.
\end{enumerate}

For $I,J \in \N ^{n+1}$, we say that $I$ and $J$ are \emph{comparable} if either $I \gg J$ or $J \gg I$.
Otherwise, they are called \emph{incomparable}.
It is clear that the order $\gg$ is not a total order on $\N^{n+1}$.
For example, $(2,0)$ and $(0,1)$ are incomparable. For a given point $I$ in $\mathbb{N}^{n+1}$, 
it is straightforward to verify that the number of points that are incomparable to $I$ is finite.


Let $S$ be a subset of $\N ^{n+1}$. An element $I \in S$ is called the \emph{greatest element of} $S$ 
if $I \gg J$ for every $J \in S \setminus \{I\}$.
By the asymmetry property of $\gg$,  the set $S$ has at most one greatest element.
This motivates the following definition.

\begin{definition} \label{def:MaximallyComparable}
An AODE $F(y)=0$ is called \emph{maximally comparable} 
if $\mathcal{E}(F)$ admits a greatest element with respect to $\gg$.
In this case, the corresponding monomial is called the \emph{highest monomial}, and the coefficient of the highest monomial 
is called the \emph{highest coefficient}.
\end{definition}

The term \emph{maximally comparable} already appeared in \cite{Vo2018}.
In \cite[Section~3]{Vo2018}, the authors defined maximally comparable first-order AODEs and studied their rational solutions.
The authors also showed that most of first-order AODEs are maximally comparable. 
Here, we extend the authors' work to the class of higher-order AODEs.
We will see later that a big part of high-order AODEs in literature are also maximally comparable.
The following theorem can be viewed as a generalization of \cite[Theorem~3.4]{Vo2018}.

\begin{theorem}\label{thm:movpol}
Let $F(y)=\sum\limits_{I \in \mathbb{N}^{n+1}} f_{I}y^{i_0} (y')^{i_1} \ldots (y^{(n)})^{i_n} \in \mathbb{K}[x]\{y\}$ 
be a differential polynomial of order $n > 0$. 
Assume that $F(y)=0$ is maximally comparable, 
and $I_0$ is the greatest element of $\mathcal{E}(F)$ with respect to $\gg$.
Then the poles of a rational solution of $F(y) = 0$ can only occur at infinity or at the zeros of $f_{I_0}(x)$.
\end{theorem}

\begin{proof}
We prove the above claim by contradiction. 
Suppose that there is $x_0 \in \mathbb{K}$ such that~$x_0$ is a pole of order $r \geq 1$ 
of a rational solution of the AODE $F(y)=0$, and $f_{I_0}(x_0) \neq 0$. Then $\ord_{x_0}f_{I_0}=0$.

We first prove that $\mathcal{M}_{x_0}(F)=\{I_0\}$. 
Since $I_0$ is the greatest element of $\mathcal{E}(F)$ with respect to $\gg$, 
we see that $||I_0|| \geq ||J||$ for all $J \in \mathcal{E}(F)$. So $I_0 \in \mathcal{D}(F)$. 
Now let us fix any $J \in \mathcal{D}(F) \setminus \{I_0\}$. Since $||I_0||=||J||$ and $||I_0||+||I||_\infty>||J||+||J||_\infty$, 
we have that $||I_0||_\infty>||J||_\infty$. 
Therefore, we conclude that $\ord_{x_0}(f_{I_0}) + ||I_0||_\infty >  \ord_{x_0}(f_J)+||J||_\infty$ 
because $\ord_{x_0}f_{I_0} = 0 \geq \ord_{x_0}(f_J)$. 
In other words, $I_0$ is the only element of $\mathcal{M}_{x_0}(F)$.

Since $\mathcal{M}_{x_0}(F) = \{I_0\}$, the indicial polynomial at $x = x_0$ has the form
$$\mathcal{P}_{x_0,F}(t) = c_{x_0}(f_{I_0}) \cdot \prod\limits_{r=0}^{n-1}{(-t-r)}^{||I_0||_{r+1}}.$$ 
It is straightforward to see that $\mathcal{P}_{x_0,F}(t)$ has no positive integer root.
Due to Proposition~\ref{PROP:LaurentSeriesSols} and $r \geq 1$, 
we have $\mathcal{E}(F) \setminus \mathcal{D}(F) \neq \emptyset$ and 
\begin{align*}
r & \leq b_{x_0}(F) = 
\max \left \{ \frac{\ord_{x_0}(f_J)+||J||_\infty-||I_0||_\infty}{||I_0||-||J||} \mathrel{\Big |} J \in \mathcal{E}(F) \setminus \mathcal{D}(F) \right \}\\
& = \max \left 
\{ 1- \frac{-\ord_{x_0}(f_J)+(||I_0||+||I_0||_\infty)-(||J||+||J||_\infty)}{||I_0||-||J||} \mathrel{\Big |} J \in \mathcal{E}(F) 
\setminus \mathcal{D}(F) \right \}\\
&<1.
\end{align*}
This contradicts the assumption that $r \geq 1$. 
\end{proof}

 The above theorem implies that for maximally comparable AODEs, 
there are only finitely many candidates for poles of rational function solutions. 
Moreover, the poles of rational functions, if there is any, occur only at the zeros of the highest coefficient with respect 
to the order $\gg$ or at infinity.
This can be considered as a generalization to nonlinear AODEs of the same fact for linear ordinary differential equations.
Once a candidate for poles of a rational solution is found, one may use Proposition~\ref{PROP:LaurentSeriesSols} 
to bound the order at this candidate. As we mentioned it (Example~\ref{exa:Nondicritical}) before, 
Proposition~\ref{PROP:LaurentSeriesSols} may fail to give the order bound at certain points, as the following example illustrates.

\begin{example} \label{EX:nobound}
Consider the following AODE:
\[
 F(y) = x^3 y y''' + x y y' - x (y')^2 + y y' = 0.
\]
It is straightforward to verify that $F(y) = 0$ is maximally comparable. 
By Theorem~\ref{thm:movpol}, we know that the poles of rational solutions of $F(y) = 0$ can only be $0$. 
However, a direct calculation implies that $\mathcal{P}_{0,F}(t) = 0$. 
Therefore, we can not give the order bound at zero by using Proposition~\ref{PROP:LaurentSeriesSols}.
\end{example}

In order to compute rational solutions of a given maximally comparable AODEs, we impose the following property to it  
so that we bound the order of candidates for poles of its rational solutions.

\begin{definition} \label{EX:cmc}
Let $F(y) = 0$ be a maximally comparable AODE with the highest coefficient $f(x)$ with respect to $\gg$. 
We say that $F(y) = 0$ is \emph{completely maximally comparable} 
if $\mathcal{P}_{x_0,F}(t)$ is a non-zero polynomial for every root $x_0$ of $f(x)$.
\end{definition}

The following is a sufficient condition for a maximally comparable AODE to be complete.

\begin{proposition} \label{PROP: cmcindicial}
Let $F(y) = 0$ be a maximally comparable AODE. 
If $\mathcal{D}(F)$ is a totally ordered set with respect to the ordering $\gg$,
then for each $x_0 \in \mathbb{K} \cup \{\infty\}$, 
we have that~$\mathcal{P}_{x_0,F}(t) \neq 0$.
\end{proposition}
\begin{proof}
Assume that $x_0 \in \mathbb{K} \cup \{\infty\}$. 
Since $F(y) = 0$ is a completely maximally comparable AODE, then for each $I, J \in \mathcal{M}_{x_0}(F)$ with $I \neq J$, 
we have that $||I||_\infty \neq ||J||_\infty$. 
On the other hand, for each $I \in \mathcal{M}_{x_0}(F)$, the degree of the polynomial $\prod\limits_{r=0}^{n-1}{(-t-r)}^{||I||_{r+1}}$ 
is exactly $||I||_\infty$. Above all, we conclude that $\mathcal{P}_{x_0,F}(t) \neq 0$.
\end{proof}

We can always give an order bound for candidates of rational solutions of 
completely maximally comparable AODEs by using Proposition~\ref{PROP:LaurentSeriesSols}. 
Combined with the partial fraction decomposition of a rational function, 
we present the following algorithm for determining all rational solutions of a completely maximally comparable AODE.

\begin{algo} \label{ALGO:ratsol}
Given a completely maximally comparable AODE $F(y) = 0$, compute all its rational function solutions.
\begin{itemize}
\item [(1)] Compute the greatest element $I_0$ of $\mathcal{E}(F)$ with respect to $\gg$. 
Compute distinct roots $x_1,\ldots,x_m$ of $f_{I_0}(x)$ in $\K$. 
\item [(2)] For $i \in \{1, \ldots, m \}$, 
compute an order bound $r_i$ for rational solutions of $F(y) = 0$ at $x = x_i$ by Proposition~\ref{PROP:LaurentSeriesSols}. 
Similarly, compute the order bound $N$ for rational solutions of the equation at infinity.
\item [(3)] Make an ansatz with the partial fraction decomposition
\begin{equation} \label{PFR}
z =\sum\limits_{i=1}^m {\sum\limits_{j=1}^{r_i}{\frac{c_{ij}}{(x-x_i)^j}}}+\sum\limits_{k=0}^N {c_ix^i} \eqspace,
\end{equation}
where the $c_{ij}$ and $c_i$ are unknown. Substitute~\eqref{PFR} into $F(y) = 0$ 
and solve the corresponding algebraic equations by using Gr\"{o}bner bases. 
\item [(4)] Return the solutions from the above step. 
\end{itemize}
\end{algo}

 The termination of the above algorithm follows from Proposition~\ref{PROP:LaurentSeriesSols}.
The correctness follows from Theorem~\ref{thm:movpol}.



\begin{example}
Consider the differential equation
\[
\begin{array}{lll}
 F(y) & = & x^2(x-1)^2y''^2+4x^2(x-1)y'y''-4x(x-1)yy'' + \\
      &   & 4x^2y'^2-8xyy'+4y^2-2(x-1)y'' \\
      & = & 0.
\end{array}
\]
We first collect some information about the exponents of terms of $F(y)$.

\begin{center}
\begin{tabular}{|c c c c c|}
\hline
$I \in \mathcal{E}(F)$ & $||I||$ & $||I||_\infty$ & $||I||+||I||_\infty$ & $f_I$\\
\hline
$(0,0,2)$ & $2$ & $4$ & $6$ & $x^2(x-1)^2$\\
$(0,1,1)$ & $2$ & $3$ & $5$ & $4x^2(x-1)$\\
$(1,0,1)$ & $2$ & $2$ & $4$ & $-4x(x-1)$\\
$(0,2,0)$ & $2$ & $2$ & $4$ & $4x^2$\\
$(1,1,0)$ & $2$ & $1$ & $3$ & $-8x$\\
$(2,0,0)$ & $2$ & $0$ & $2$ & $4$\\
\hline
$(0,1,0)$ & $1$ & $1$ & $2$ & $-2(x-1)$\\
\hline
\end{tabular}
\end{center}

In the above table, $\mathcal{D}(F)$ consists of the first $6$ elements of $\mathcal{E}(F)$, and $d(F)=2$. 
The first one $(0,0,2)$ is the greatest element of $\mathcal{E}(F)$ with respect to $\gg$. 

By Theorem~\ref{thm:movpol}, the poles of a rational solution of $F(y) = 0$ 
can only occur at the zeros of the polynomial $x^2(x-1)^2$, which are $0$ and $1$, and probably at infinity.

A simple computation based on Proposition~\ref{PROP:LaurentSeriesSols} shows that the orders of poles of a rational 
solution of $F(y) = 0$ at $0, 1$, and 
infinity are at most $0,1$, and $1$, respectively. 

Hence, we make an ansatz of the form:  
$$z = \frac{c_1}{x-1}+c_2+c_3x \ \ \text{ for some } \ \ c_1,c_2,c_3 \in \mathbb{K}.$$
Substituting $z$ into $F(y) = 0$ and solving the corresponding algebraic equations, 
we find that the rational solutions of $F(y) = 0$ are $c_3 x$ and $\frac{1}{x-1}+c_3 x$, 
where $c_3$ is an arbitrary constant in $\K$.
\end{example}

%

\section{Experimental results} \label{sec:ExperimentalResults}


Some new classes of AODEs have been introduced through previous sections based on properties of 
their polynomial and rational solutions. They are: noncritical, maximally comparable and completely maximally comparable AODEs.
%
In this section, we do some statistical investigation for noncriticality and 
the (completely) maximal comparability of AODEs 
from the famous collection of differential equations by Kamke \cite{Kamke}. 
The corresponding \texttt{Maple} worksheet is available in:

\begin{center}
 \href{https://yzhang1616.github.io/KamkeODEs.mw}{https://yzhang1616.github.io/KamkeODEs.mw}
\end{center}

The worksheet requires the availability of the following \texttt{Maple} package:

\begin{center}
 \href{https://yzhang1616.github.io/KamkeODEs.mpl}{https://yzhang1616.github.io/KamkeODEs.mpl}
\end{center}

There are 834 AODEs in Kamke's collection. All of them are noncritical. 
It means that our method can be used to determine all polynomial solutions, if there is any, 
of each AODE from Kamke's collection. Among them, there are 655 maximally comparable AODEs ($\approx$ 78.54 \%).  
All of the maximally comparable AODEs are complete.

The class of AODEs covers around 79.66 \% of the entire collection of ODEs. 
The remaining ODEs have coefficients involving trigonometric functions ($\sin x, \cos x$,...), 
hyperbolic functions ($\sinh x, \cosh x$, ...),
exponential functions $e^x$, logarithmic functions $\log x$, or power functions with parameters in the exponents 
($x^\alpha, y^\beta$, ...). For certain choices of the parameters, 
the latter ODEs will become algebraic. More precisely, 
there are 35 ODEs containing parameters in the power functions. If the parameters are chosen in a 
suitable way such that the corresponding ODEs are algebraic, then all of them are noncritical and 21 among them (60 \%) 
are completely maximally comparable.

\section*{Acknowledgement}
We thank Matteo Gallet and Christoph Koutschan for valuable suggestions on revising our paper. 

\bibliographystyle{abbrv}

\end{document}